\documentclass[conference]{IEEEtran}
\IEEEoverridecommandlockouts

\usepackage{cite}
\usepackage{amsmath,amssymb,amsfonts}
\usepackage{algorithmic}
\usepackage{graphicx}
\usepackage{physics}
\usepackage{tikz}
\usetikzlibrary{arrows.meta,positioning,fit,calc}
\usepackage{textcomp}
\usepackage{xcolor}
\usepackage{stfloats}
\usepackage[section]{placeins}

\usepackage{amsthm}

\theoremstyle{plain}
\newtheorem{theorem}{Theorem}

\theoremstyle{remark}
\newtheorem{remark}{Remark}

\def\BibTeX{{\rm B\kern-.05em{\sc i\kern-.025em b}\kern-.08em
    T\kern-.1667em\lower.7ex\hbox{E}\kern-.125emX}}

\begin{document}

\title{Channel Capacity under the Subtractive Dithered
Quantization Model\\}

\author{
\IEEEauthorblockN{Hossein Atrsaei\textsuperscript{*}, Mireille Sarkiss\textsuperscript{\textdagger}, Mich\`ele Wigger\textsuperscript{\S}}

\IEEEauthorblockA{\textsuperscript{*}LTCI, T\'el\'ecom Paris, Institut Polytechnique de Paris, Palaiseau, France\\
Email: hossein.atrsaei@ip-paris.fr}

\IEEEauthorblockA{\textsuperscript{\textdagger}SAMOVAR, T\'el\'ecom SudParis, Institut Polytechnique de Paris, Palaiseau, France\\
Email: mireille.sarkiss@telecom-sudparis.eu}

\IEEEauthorblockA{\textsuperscript{\S}Universit\'e Paris-Saclay, CNRS, CentraleSup\'elec, Laboratoire des Signaux et Syst\`emes, 91190 Gif-sur-Yvette, France\\
Email: michele.wigger@centralesupelec.fr}
}
\sloppy

\maketitle

\begin{abstract}
We study the capacity of an additive white Gaussian noise (AWGN) channel followed by a subtractive dithered uniform quantizer. Under the Schuchman conditions and with negligible overload probability, the system admits an additive-noise representation in which the effective noise is the sum of Gaussian and uniform components.

Capacity bounds are derived for this model where inputs are subject to an average-power constraint as well as a peak-amplitude constraint, where the latter accounts for the limited quantizer dynamic range. Specifically, a computable lower bound is obtained based on the entropy power inequality (EPI), using the maximum-entropy input under the above constraints. Tighter numerical lower bounds are  derived using discrete input constellations with finite mass points. Finally, an upper bound is obtained by exploiting the maximum-entropy property of the Gaussian distribution for a given variance.

Numerical results show that, for a \(K\)-level quantizer, discrete constellations with \(K\) mass points already achieve near-optimal rates among the tested families. Moreover, our upper bound is close to the lower bounds in the moderate signal-to-noise ratio (SNR) regime; thus it provides a simple capacity approximation in this regime. 

\end{abstract}



\section{Introduction}

Modern communication systems rely on digital signal processing, where analog-to-digital converters (ADCs) convert the received analog waveform into a finite-resolution digital signal from which the transmitted message is reconstructed and decoded. Low-precision ADCs, operating with 1–3 bits per sample, have attracted significant interest due to their reduced power consumption and hardware complexity~\cite{walden1999, schreier2005}, particularly in millimeter-wave (mmWave) and massive multiple-input multiple-output (MIMO) systems~\cite{liang2016}.

Understanding the fundamental information-theoretic limits imposed by coarse quantization on channel capacity is therefore a problem of both theoretical and practical importance. Closed-form expressions for channel capacity are generally intractable, and the literature has focused on two main directions. One line of research studies analytically tractable cases, most notably the one-bit ADC model. For an additive white Gaussian noise (AWGN) channel followed by a sign quantizer, Singh, Dabeer, and Madhow showed that antipodal signaling achieves capacity under an average-power constraint~\cite{singh2009limits}. Another line of work derives computable bounds for more general quantization models. For example, \cite{dutta2020} treats quantization as additive Gaussian noise in the large-system limit. Such models, including the additive quantization noise model (AQNM)~\cite{fletcher2007,orhan2015}, are approximate and may become inaccurate in the low-resolution regime.

Smith~\cite{smith1971} established that, for an AWGN channel under simultaneous peak-amplitude and average-power constraints on the channel input, the capacity-achieving input distribution is discrete with a finite number of mass points. Extensions of this result to broader channel models suggest that discreteness persists under fairly general conditions on the noise distribution~\cite{tchamkerten2004discreteness, fahs2018support, elmoslimany2018discreteness}.

For quantized-output channels, Singh et al.~\cite{singh2009limits} further showed that an average-power constraint induces bounded input support and that, for a quantizer with \(K\) output levels, a symmetric input with at most \(K+1\) mass points suffices.

In this paper, we adopt a model based on \emph{subtractive dithered quantization}. A dither is added before quantization and subtracted afterward, which reduces the
input-output relation to \(Y=X+\bar W\), where \(\bar W\) is the sum of Gaussian channel noise and uniform quantization noise. This holds under the Schuchman conditions~\cite{gray1993dithered,schuchman1964}, provided that the dither is independent of the input and properly distributed, and that the quantizer operates without overload. Under these conditions, the quantization error is uniformly distributed and independent of the input.

To ensure that the quantizer operates in its linear regime \([-\gamma, \gamma]\) with high probability, we impose a \emph{peak-amplitude constraint} \(|X| \le A\), so that the overload probability is at most \(\varepsilon\), for a small constant \(\varepsilon > 0\). 

Under this framework, we derive bounds on the capacity. Our upper bound follows from the variance of the output signal together with the maximum-entropy property of the Gaussian distribution under a fixed second moment. 
We present the following lower bounds:
\begin{itemize}
    \item an entropy-power-inequality (EPI)-based bound, derived from the maximum-entropy input under the peak-amplitude and average-power constraints; and
    \item discrete-input lower bounds obtained by numerically optimizing symmetric constellations with \(N=2\), \(3\), \(4\), \(5\), \(8\), and \(9\) mass points.
\end{itemize}
Numerical results show that the proposed upper and lower bounds are tight at moderate signal-to-noise ratios (SNRs). The EPI lower bound and the variance-based upper bound therefore provide simple numerical approximations to capacity in this regime. Across all SNR regimes, our bounds suggest that, for a $K$-level quantizer, an $N=K$-point input distribution achieves a rate close to capacity; further increasing the number of mass points hardly improves the achievable rates. At low SNR, two-point constellations seem sufficient. 

Note that the effective noise in our channel (consisting of Gaussian and uniform noise components) has a continuous density with finite variance and satisfies regularity conditions as stated in~\cite{smith1971, fahs2018support, tchamkerten2004discreteness}. By these results, the capacity-achieving input is discrete with finite support at all SNR levels. 

\section{System Model}

We consider a real, discrete-time AWGN channel
\begin{equation}
    Z = X + W,
    \label{eq:channel_pre_quant}
\end{equation}
where \(W \sim \mathcal{N}(0, \sigma^2)\) is independent of \(X\).

The receiver employs a \(K\)-level subtractive dithered uniform quantizer with dynamic range \([-\gamma,\gamma]\) and step size $\Delta = {2\gamma}/{K}$. Here, \(K\) denotes the number of quantizer output levels; hence, a \(b\)-bit ADC has \(K=2^b\) output levels. For subtractive dithering, an independent dither \(D \sim \mathcal{U}(-\Delta/2,\Delta/2)\) is added before quantization and subtracted afterward, as illustrated in Fig.~\ref{fig:system_model}.
\begin{figure*}[t]
\centering
\begin{tikzpicture}[
    >=Latex,
    font=\footnotesize,
    node distance=0.55cm,
    block/.style={draw, rectangle, minimum height=5mm, minimum width=10mm, align=center},
    sum/.style={draw, circle, inner sep=0pt, minimum size=4.5mm}
]
\node (m) {$M$};
\node[block, right=of m]   (enc)  {Encoder};
\node[sum,   right=of enc] (sumw) {$+$};
\node[sum,   right=of sumw](sumd) {$+$};
\node[block, right=of sumd](adc)  {$q_{K}(\cdot)$};
\node[sum,   right=of adc] (subd) {$-$};
\node[block, right=of subd](dec)  {Decoder};
\node[right=of dec] (mhat) {$\hat{M}$};

\draw[->] (m)    -- (enc);
\draw[->] (enc)  -- node[above] {$X$} (sumw);
\draw[->] (sumw) -- node[above] {$Z$} (sumd);
\draw[->] (sumd) -- node[above] {$U$} (adc);
\draw[->] (adc)  -- (subd);
\draw[->] (subd) -- node[above] {$Y$} (dec);
\draw[->] (dec)  -- (mhat);

\node[above=0.55cm of sumw] (w) {$W$};
\draw[->] (w) -- (sumw);

\node[below=0.45cm of sumd] (d1) {$D$};
\node[below=0.45cm of subd] (d2) {$D$};
\draw[->] (d1) -- (sumd);
\draw[->] (d2) -- (subd);

\node[draw, dashed, inner sep=5pt,
      fit=(sumd)(adc)(subd)(d1)(d2)] (qbox) {};
\node[below=3pt of qbox, font=\scriptsize]
      {Subtractive Dithered Quantizer};
\end{tikzpicture}
\caption{Block diagram of the communication system model with subtractive
dithered quantization at the receiver.}
\label{fig:system_model}
\end{figure*}
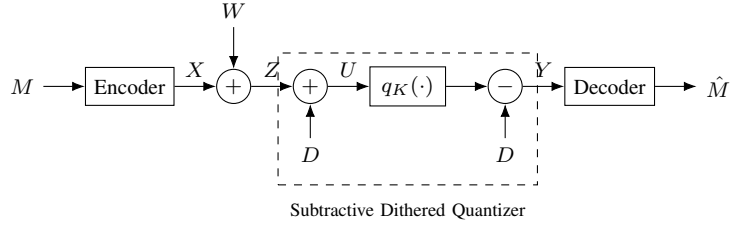
For quantization, the interval \([-\gamma,\gamma]\) is partitioned into \(K\) uniform subintervals of width \(\Delta\), whose midpoints define the reconstruction levels, with saturation at \(\pm(\gamma-\Delta/2)\) outside the range.

Under the Schuchman conditions for subtractive dithering~\cite{gray1993dithered,schuchman1964}, and assuming that the probability of overload outside \([-\gamma,\gamma]\) is sufficiently small, the overall input–output relation can be modeled as
\begin{equation}
    Y = X + \bar{W}, \label{eq:additive_model}
\end{equation}
where the effective noise is 
\begin{equation}\bar{W} = W + W_q,\label{eq:Wbar}
\end{equation} with \(W_q \sim \mathcal{U}(-\Delta/2,\,\Delta/2)\) independent of both \(X\) and \(W\). The probability density function (PDF) and cumulative distribution function (CDF) of \(\bar{W}\), denoted by \(f_{\bar{W}}(\cdot)\) and \(F_{\bar{W}}(\cdot)\), are provided in Appendix~\ref{app:noise_dist}.

\subsection{Input Constraints}

Let \(U=X+W+D\) denote the quantizer input. The restriction of
the overload probability to be small directly limits the admissible
channel inputs. We therefore impose the pointwise overload constraint
\begin{equation}
    \Pr\bigl(|U|>\gamma \mid X=x\bigr)\leq \varepsilon,
    \qquad \forall\,x\in\operatorname{supp}(P_X),
    \label{eq:overload_pointwise}
\end{equation}
for a given tolerance \(0<\varepsilon<1\). Since \(D\) has
the same distribution as \(W_q\), the overload probability is
\begin{equation}
    \Pr\bigl(|U|>\gamma \mid X=x\bigr)
    =
    1-\bigl[
        F_{\bar W}(\gamma-x)-F_{\bar W}(-\gamma-x)
    \bigr].
    \label{eq:overload_prob}
\end{equation}
By the symmetry and strict unimodality of \(f_{\bar W}\), this probability is continuous and strictly increasing in \(|x|\). Consequently, the supremum over the support is attained at the boundary \(x = A\). It therefore suffices to enforce the constraint at \(x = A\), which results in the implicit equation
\begin{equation}
    1-\bigl[
        F_{\bar W}(\gamma-A)-F_{\bar W}(-\gamma-A)
    \bigr]
    =
    \varepsilon.
    \label{eq:A_exact}
\end{equation}
Accordingly, we consider input distributions \(P_X\) satisfying the
\emph{peak-amplitude constraint}
\begin{equation}
    |X|\leq A
    \label{eq:peak_constraint}
\end{equation}
and the \emph{average-power constraint}
\begin{equation}
    \mathbb{E}[X^2]\leq P_0.
    \label{eq:avg_constraint}
\end{equation}

\begin{remark}
Let \( \varepsilon_t \triangleq 2\bigl(1-F_{\bar W}(\gamma)\bigr) \) denote the overload probability at \(A=0\). Since the left-hand
side of~\eqref{eq:A_exact} is continuous and strictly increasing
in \(A\), a unique positive solution exists if and only if
\(\varepsilon>\varepsilon_t\). Moreover,
\[
\Pr\bigl(|U|>\gamma\mid X=\gamma\bigr)
=
\frac{1}{2}+F_{\bar W}(-2\gamma)
>
\frac{1}{2},
\]
where \(F_{\bar W}(0)=1/2\) follows from the symmetry of
\(f_{\bar W}\). Hence, for
\(\varepsilon_t<\varepsilon\leq 1/2\), the unique solution lies
in \((0,\gamma)\) and is computed numerically by bisection over
\([0,\gamma]\).
\end{remark}

\subsection{Problem Formulation}

We study the channel capacity under the input constraints in \eqref{eq:peak_constraint} and \eqref{eq:avg_constraint}. The capacity, denoted \(C(A,P_0)\), is defined as
\begin{equation}
    C(A,P_0) = \sup_{P_X \in \mathcal{F}} I(X;Y),
    \label{eq:capacity_def}
\end{equation}
where the feasible set is given by
\begin{equation}
    \mathcal{F} = \left\{ P_X : \mathbb{E}[X^2] \le P_0,\; |X| \le A \right\},
\end{equation}
and \(I(X;Y)\) denotes the mutual information between the channel input and output.

This class of capacity problems for additive-noise channels under input constraints has been studied extensively (see, e.g.,~\cite{smith1971,fahs2018support}). However, explicit characterizations are generally unavailable when the additive noise is non-Gaussian and the amplitude constraint is induced by the finite dynamic range of the quantizer.

Under the additive-noise model in~\eqref{eq:additive_model}, the mutual information can be expressed as
\begin{equation}
    I(X;Y) = h(Y) - h(\bar{W}),
    \label{eq:MI_entropy}
\end{equation}
where \(h(\cdot)\) denotes differential entropy.

By symmetry of the noise and the constraint set, a capacity-achieving input can be chosen symmetric, which can be proved analogously to~\cite[Lemma~2]{singh2009limits}. Therefore, we restrict attention to symmetric inputs.

\section{Theoretical Results}
\label{sec:theoretical}

Finding a closed-form expression for the capacity is intractable. We derive upper and lower bounds on \(C(A,P_0)\).

\subsection{Upper Bounds on Channel Capacity}

Because a Gaussian distribution maximizes differential entropy for a given variance, applying this property to~\eqref{eq:MI_entropy} yields the following upper bound.

\begin{theorem}[Variance-Based Upper Bound]
\label{thm:variance_ub}
The capacity of our model is upper bounded as
\begin{IEEEeqnarray}{rCl}
    C(A,P_0)
    &\leq&
    \frac{1}{2}
    \log\!\Bigg(
    2\pi e
    \Big(
    \min\{A^2,P_0\}
    +\sigma^2+\frac{\Delta^2}{12}
    \Big)
    \Bigg)
    \nonumber\\[-0.5ex]
    &&-h(\bar{W}),
\label{eq:capacity_ub_variance_main}
\end{IEEEeqnarray}
where \(h(\bar{W}):= -\int_{-\infty}^{\infty}
f_{\bar{W}}(w)\log f_{\bar{W}}(w)\,dw \) is the noise differential entropy, evaluated numerically.
\end{theorem}


In the low-power regime \(P_0 \to 0\), the upper bound converges to \(\frac{1}{2}\log\!\left(2\pi e\!\left(\sigma^2+\frac{\Delta^2}{12}\right)\right)
- h(\bar{W})\), which is strictly positive; hence, the bound is loose in this regime.

In the moderate-power regime, the upper bound is observed to be close to the actual capacity and provides a simple estimate in this range.

In the high-power regime $P_0\to \infty$, the upper bound saturates at  $\frac{1}{2}\log \left( 2\pi e \left(A^2+\sigma^2+\frac{\Delta^2}{12}\right)\right)- h(\bar{W})$, proving that the capacity also saturates at a finite value.

A second upper bound on \(C(A,P_0)\) is obtained by relaxing the peak-amplitude constraint~\eqref{eq:peak_constraint} and removing quantization. This reduces the model to an unquantized-AWGN channel subject only to the average-power constraint. Therefore,
\begin{equation}
    C(A,P_0)
    \leq
    \frac{1}{2}
    \log\!\left(
    1+\frac{P_0}{\sigma^2}
    \right).
\label{eq:capacity_ub_awgn}
\end{equation}
Unlike the variance-based bound, this one vanishes as \(P_0 \to 0\) and is therefore tighter at low SNR, though it does not capture the high-SNR saturation.

\subsection{Lower Bounds}

\subsubsection{An EPI-Based Lower Bound}
The EPI yields the following bound.

\begin{theorem}[EPI-Based Lower Bound]
\label{thm:epi_lb}
The capacity of our model is lower bounded as
\begin{equation}
    C(A,P_0)
    \ge
    \frac{1}{2}
    \log\!\left(
    e^{2h_{\max}(A,P_0)} + e^{2h(\bar{W})}
    \right)
    - h(\bar{W}),
    \label{eq:capacity_LB_EPI_main}
\end{equation}
where
\begin{equation}
    h_{\max}(A,P_0)
    :=
    \sup_{P_X \in \mathcal{F}} h(X).
    \label{eq:hmax}
\end{equation}

The maximization in~\eqref{eq:hmax} leads to the following expression:
\begin{equation}
    h_{\max}(A,P_0)
    =
    \begin{cases}
    \log(2A), & P_0 \ge \dfrac{A^2}{3}, \\[1ex]
    -\log c + \mu^\star P_0, & P_0 < \dfrac{A^2}{3},
    \end{cases}\label{eq:hmax2}
\end{equation}
with
\begin{equation}
    c =
    \frac{\sqrt{\mu^\star}}
    {\sqrt{\pi}\Bigl(1 - 2Q\!\bigl(\sqrt{2\mu^\star}\,A\bigr)\Bigr)},
    \label{eq:c_def}
\end{equation}
where \(\mu^\star > 0\) solves
\begin{equation}
    \frac{1}{2\mu^{\star}}
    -
    \frac{A e^{-\mu^{\star} A^2}}
    {\sqrt{\pi \mu^{\star}}\Bigl(1 - 2Q\!\bigl(\sqrt{2\mu^{\star}}\,A\bigr)\Bigr)}
    =
    P_0.
    \label{eq:mu_star_theorem}
    \end{equation}

\end{theorem}

\begin{IEEEproof}
See Appendix~\ref{app:epi_lb}.
\end{IEEEproof}

In the low-power regime \(P_0\to0\), we have \(\mu^\star\approx1/(2P_0)\), and hence \(\mu^\star P_0\approx1/2\). Moreover, since \(\mu^\star\to\infty\),
\(c\approx1/\sqrt{2\pi P_0}\). 
Substituting these approximations into \eqref{eq:hmax2}
and \eqref{eq:capacity_LB_EPI_main} gives the low-power approximation
\begin{equation}
    C(A,P_0)
    \gtrsim
    \frac{1}{2}
    \log\!\left(1 + 2\pi e\,P_0\, e^{-2h(\bar{W})}\right).
\end{equation}
This lower bound approaches zero linearly with \(P_0\) (since \(\log(1+x)\approx x\) for small \(x\)), similarly to the unquantized-AWGN upper bound in~\eqref{eq:capacity_ub_awgn}; see the top panel of Fig.~\ref{fig:bounds_lowhigh_K4}.

In the moderate-SNR regime, our numerical results show that the EPI lower bound is close to the variance-based upper bound and thus provides a tractable approximation of capacity.

At high power \(P_0 \geq A^2/3\), the EPI-based lower bound saturates at $\frac{1}{2} \log\left(4A^2 + e^{2h(\bar{W})}\right)-h(\bar{W})$. The variance-based upper bound also saturates, although at a higher value; see the numerical results in the next section.

\subsubsection{Discrete-Input Lower Bounds}
 
We obtain improved lower bounds by restricting the input distribution to symmetric finite-support constellations with \(N\) points. For \(N=2\), we use \(\mathcal X=\{-a,a\}\). For \(N\ge3\), we consider
\begin{equation}
    \mathcal{X} =
    \begin{cases}
    \{0,\pm a\beta_1,\ldots,\pm a\beta_{\frac{N-3}{2}},\pm a\},
        & N \text{ odd},\\
    \{\pm a\beta_1,\ldots,\pm a\beta_{\frac{N-2}{2}},\pm a\},
        & N \text{ even},
    \end{cases}
    \label{eq:constellation_general}
\end{equation}
where \(0\le \beta_1\le\cdots\le1\), \(0 < a\le A\), and the probabilities are symmetric.

For any such constellation, with support points \(\{x_n\}_{n=1}^N\) and probabilities \(\{p_n\}_{n=1}^N\), the power constraint becomes
\begin{equation}
    \mathbb{E}[X^2]
    = a^2\Psi(\boldsymbol{\beta},\mathbf p),
    \qquad
    \Psi(\boldsymbol{\beta},\mathbf p)
    := \sum_{n=1}^N p_n \left(\frac{x_n}{a}\right)^2.
    \label{eq:psi_def}
\end{equation}
Hence, for fixed \((\boldsymbol{\beta},\mathbf p)\), the admissible amplitude satisfies
\begin{equation}
    a \le
    \min\left\{
        \sqrt{\frac{P_0}{\Psi(\boldsymbol{\beta},\mathbf p)}}, A
    \right\}.
    \label{eq:amplitude_bound}
\end{equation}
We evaluate 2-, 3-, 4-, 5-, 8-, and 9-point constellations as 
special cases of~\eqref{eq:constellation_general}. For each family, \(I(X;Y)\) is maximized numerically over \((\boldsymbol{\beta},\mathbf p,a)\) subject to the peak-amplitude and average-power constraints, leading to a sequence of increasingly tight lower bounds on the capacity \(C(A,P_0)\).

\begin{remark}
The effective noise $\bar{W}$ has a continuous density with finite variance and admits an analytic extension. Under such regularity conditions, 
the capacity-achieving input distribution is discrete with finite support; see classical results for amplitude-constrained channels~\cite{smith1971, fahs2018support, tchamkerten2004discreteness}.
\end{remark}

\section{Numerical Results}

We evaluate the bounds numerically for \(\sigma^2 = 10^{-2}\),
\(\gamma = 2\), \(\varepsilon = 10^{-4}\), and different values of \(K\), with $\mathrm{SNR} \triangleq P_0/\sigma^2$. The amplitude constraint \(A\) is obtained from \eqref{eq:A_exact}. For \(K \in \{2, 4, 8\}\), \(A\) is weakly affected by \(\varepsilon\) and remains nearly unchanged, for example when \(\varepsilon\) is reduced to \(10^{-6}\).


The choice $\gamma=2$ is motivated by practical considerations as explained in~\cite{walden1999, schreier2005}. To illustrate the theoretical impact of the dynamic range, we also present results for \(\gamma \in \{5, 10\}\).

\begin{figure}[!t]
    \centering
    \includegraphics[width=\columnwidth]{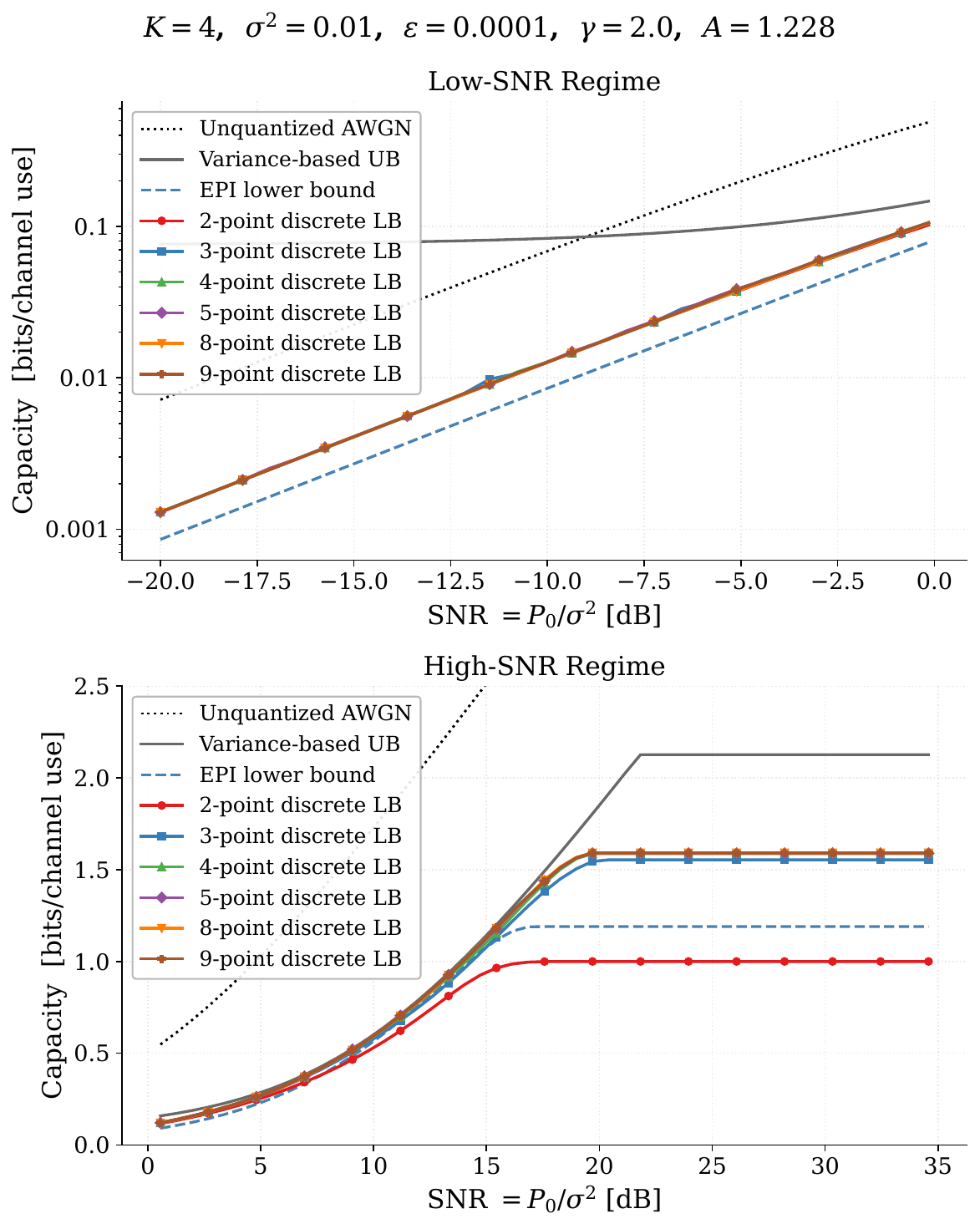}
    \caption{%
        Capacity bounds for a quantizer with \(K=4\) output levels. \textit{Top}: low-SNR regime on a logarithmic scale.
        \textit{Bottom}: high-SNR regime.
    }
    \label{fig:bounds_lowhigh_K4}
\end{figure}
 
\subsection{Capacity Bounds for \(K = 4\)}
 
Fig.~\ref{fig:bounds_lowhigh_K4} shows the capacity bounds derived in the previous section for $K = 4$.

In the moderate- and high-SNR regimes, the numerical results confirm the analytical observations of Section~\ref{sec:theoretical}. Our variance-based upper bound significantly improves over the unquantized-AWGN upper bound, and seems tight in the moderate-SNR regime; its high-SNR saturation level, however, lies considerably above all lower bounds.

The EPI lower bound is close to the variance-based upper bound at moderate SNR but degrades relative to the discrete-input lower bounds at high SNR. Among the discrete-input distributions, the lower bound increases as \(N\) grows from \(2\) to \(4\). Increasing \(N\) beyond \(K = 4\) produces negligible additional gain for these parameters.

As the SNR approaches zero, the variance-based upper bound converges to a strictly positive value and is therefore looser than the unquantized Gaussian capacity in this regime. A gap between the best upper and lower bounds persists throughout. In the low-SNR regime, the two-point lower bound performs best, and increasing \(N\) beyond 2 does not lead to any noticeable improvement. It also strictly outperforms the EPI lower bound.

\begin{figure}[!t]
    \centering
    \includegraphics[width=\columnwidth]{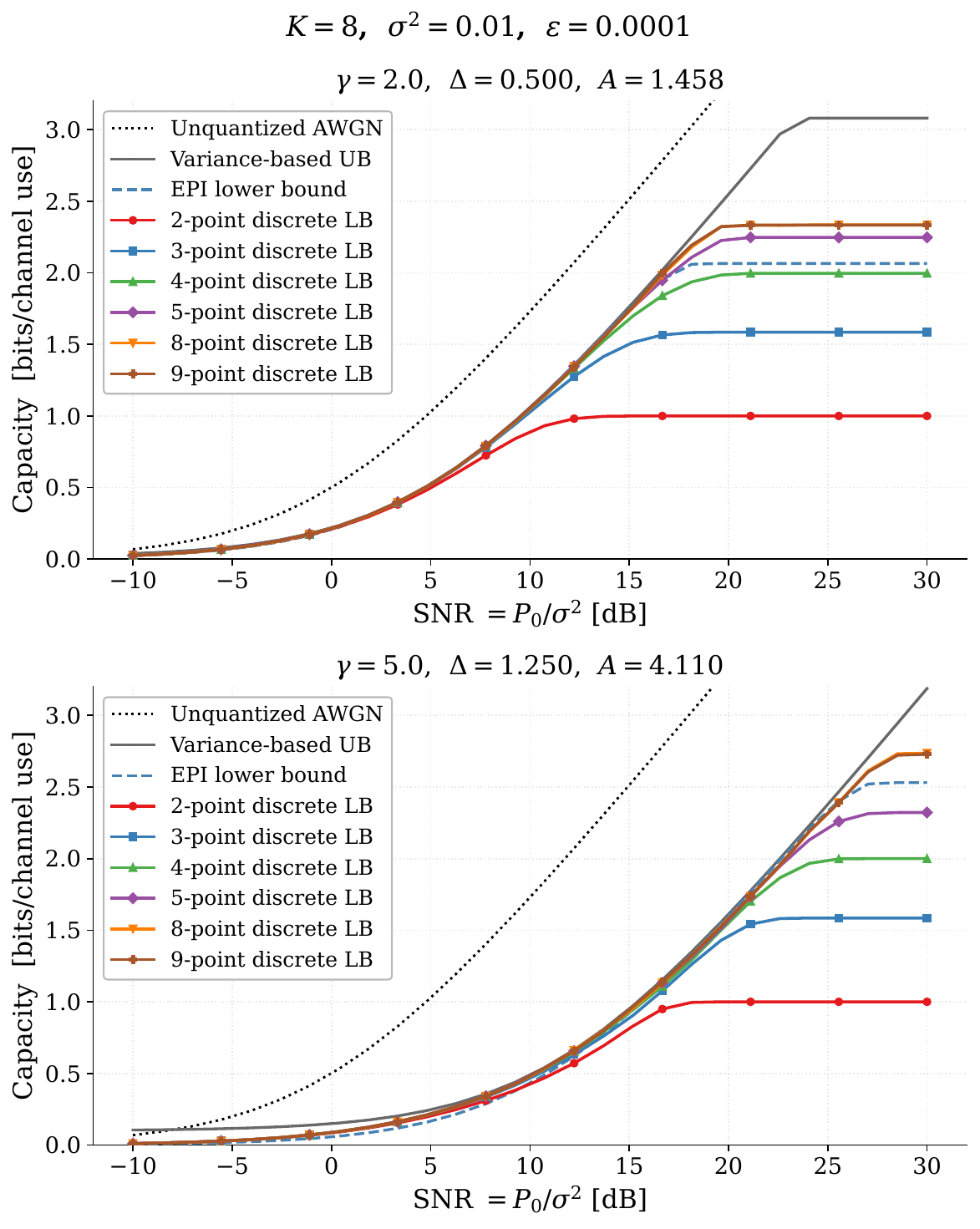}
    \caption{%
        Capacity bounds for a quantizer with \(K=8\) output levels and two values of the ADC dynamic range:
        \(\gamma = 2\) (top) and \(\gamma = 5\) (bottom).
    }
    \label{fig:bounds_full_K8}
\end{figure}
 
\subsection{Capacity Bounds for \(K = 8\)}
\label{sec:K8_bounds} 
Fig.~\ref{fig:bounds_full_K8} shows the capacity bounds for \(K = 8\) and \(\gamma \in \{2,5\}\). The behavior is similar to that observed for \(K=4\), and, for the same value of $\gamma$, the evaluated bounds are generally higher for $K=8$. As observed for \(K=4\), the 9-point constellation yields no 
noticeable improvement over the 8-point one.

Our bounds further allow us to conclude that larger values of \(\gamma\) reduce the capacity in the low-SNR regime but increase it in the high-SNR regime. This occurs because increasing \(\gamma\) increases the quantization-noise variance, degrading performance at low SNR, while also increasing the admissible amplitude \(A\), which becomes the limiting factor at high SNR.

Consistently, for $\gamma=5$, each $N$-point lower bound   approaches $\log_2 N$~bits/channel use in the high-SNR regime. This occurs because a larger admissible amplitude $A$ allows the $N$ constellation points to be well separated, which results in a small conditional entropy $H(X|Y)$. Furthermore, at high SNR, the input distribution can be chosen nearly uniform over the constellation, so that $I(X;Y)= H(X)-H(X|Y)\approx H(X)\approx\log_2(N)$.

\begin{figure}[t]
    \centering
    \includegraphics[width=\columnwidth]{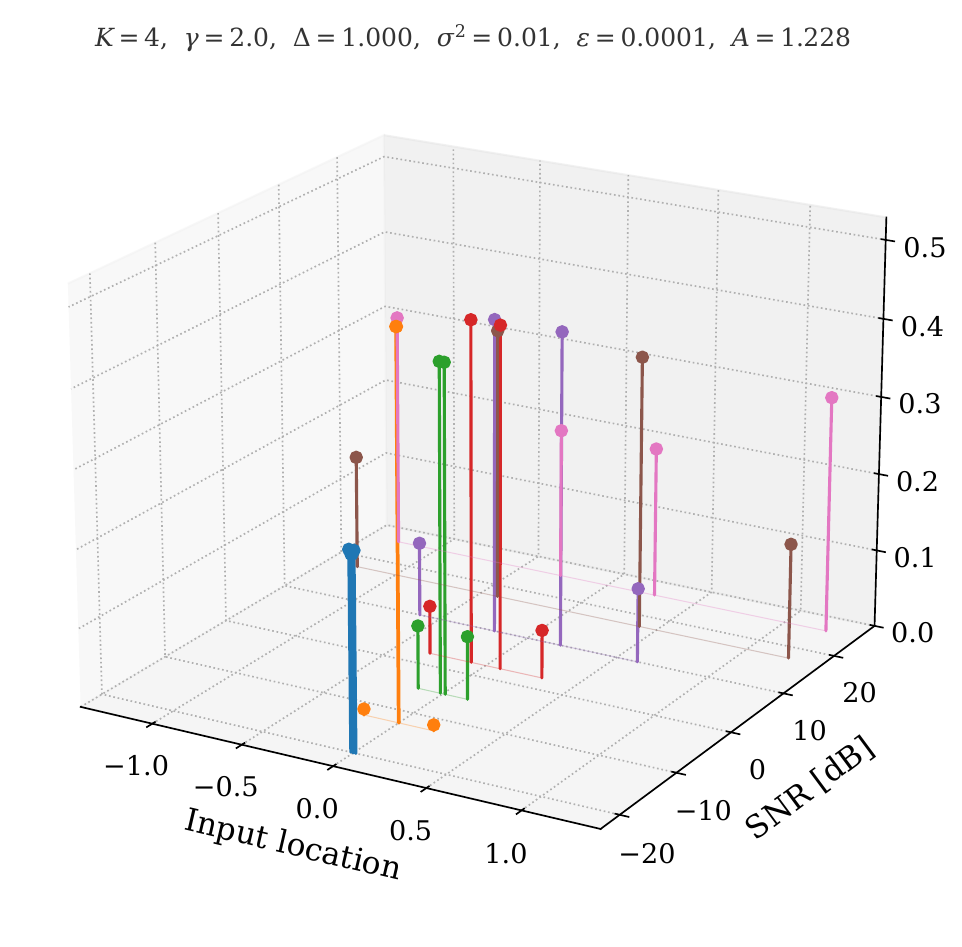}
    \caption{%
      The optimal 4-point-input-distribution, as a function of SNR.
        Each stem represents a mass point, and its height indicates the corresponding probability.
    }
    \label{fig:dist3d_K4}
\end{figure}
 
\subsection{The Maximizing Input Distribution as a Function of SNR}

Fig.~\ref{fig:dist3d_K4} shows the optimal 4-point-input-distribution as a function of SNR. As SNR increases, the distribution evolves from one concentrated near the origin to one that allocates mass approximately uniformly across the support $|x|\leq A$. This reflects the transition from a \emph{noise-limited regime}, where Gaussian-like inputs are favorable, to a \emph{constraint-limited regime}, where maximizing the input entropy under the constraint $|X|\le A$ becomes dominant.

\begin{figure}[t]
    \centering
    \includegraphics[width=\columnwidth]{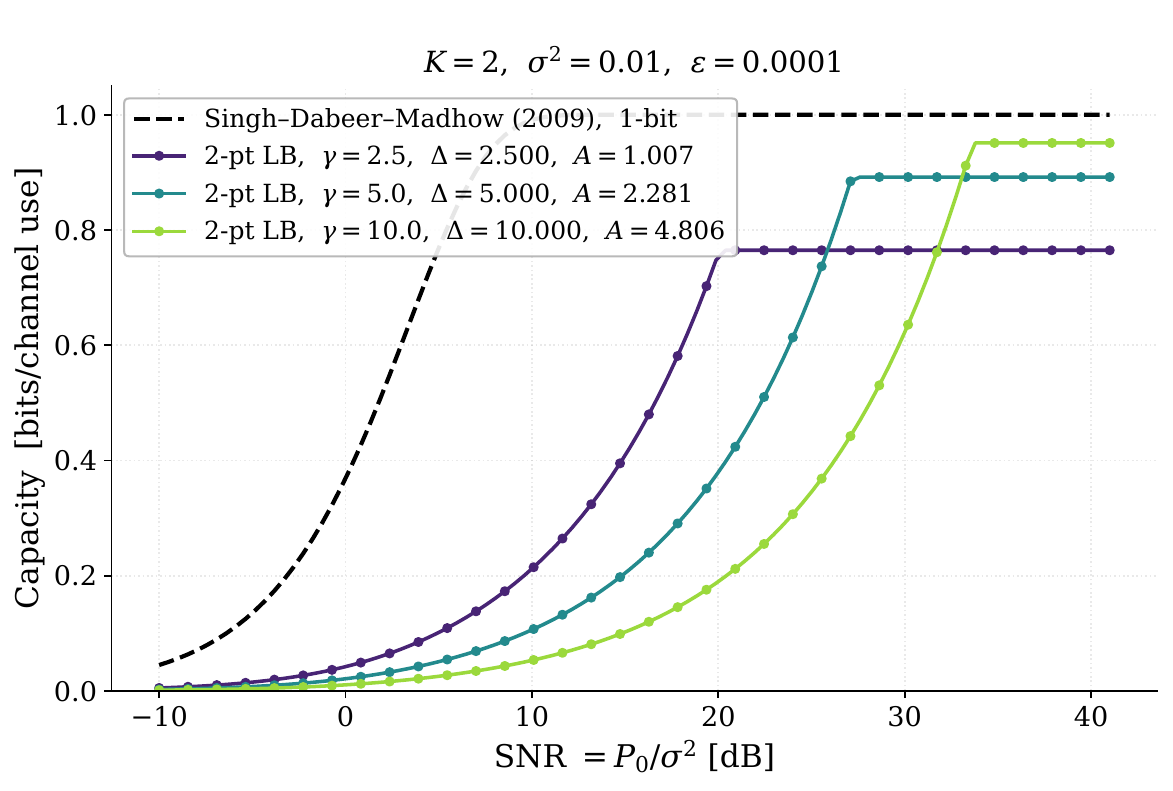}
    \caption{%
    Comparison between the capacity of the one-bit sign-quantizer model in~\cite{singh2009limits} and the 2-point discrete-input lower bound for the \(K=2\) subtractive-dithered model at different values of \(\gamma\).
    }
    \label{fig:singh_comparison}
\end{figure}

\subsection{Comparison with~\cite{singh2009limits} (\(K = 2\))}
Fig.~\ref{fig:singh_comparison} compares the 2-point achievable rate of the \(K=2\) subtractive-dithered model with the one-bit sign-quantizer capacity in~\cite{singh2009limits}. In our model, the peak-amplitude constraint controls overload and is set by the quantizer range and the overload tolerance. The model in~\cite{singh2009limits} imposes only an average-power constraint and does not consider a finite dynamic range.

At low SNR, the gap between the two models is mainly due to the dither noise. At high SNR, increasing \(\gamma\) relaxes the peak-amplitude constraint in our model. The achievable rate then approaches \(1\) bit/channel use as the two antipodal points become well separated, matching the capacity of the sign-quantizer model in~\cite{singh2009limits}. However, for small values of \(\gamma\), the gap remains substantial and reflects the stringent peak-amplitude constraint imposed by the finite dynamic range in our model.


 
 

\section{Conclusions}

This paper studied the capacity of a real AWGN channel with a \(K\)-level subtractive dithered quantizer. By the dithering principle~\cite{gray1993dithered,schuchman1964}, the channel is represented as an additive-noise channel whose effective noise is the sum of Gaussian and uniform components, with the input subject to joint peak-amplitude and average-power constraints. The peak-amplitude constraint arises from the finite dynamic range of the ADC and is determined by the number of reconstruction levels \(K\), the quantizer range \(\gamma\), and the overload tolerance \(\varepsilon\).

We derived an EPI-based lower bound on the capacity using a truncated Gaussian input, tighter numerical lower bounds using inputs with a finite number of mass points, and a variance-based upper bound.

Our numerical comparison of the above bounds shows that
the bounds are tight in the moderate-SNR regime, where the variance-based upper bound and the EPI lower bound thus provide simple approximate characterizations of capacity. In both the low- and high-SNR regimes, the EPI bound is strictly below the bounds obtained with a finite number of mass points, and the upper and lower bounds do not coincide. 

Another observation is that a discrete input constellation with \(N=K\) points approaches the best rates obtained among the tested constellations. At low SNR, a 2-point constellation performs nearly as well as the larger constellations considered, irrespective of the number of quantizer output levels \(K\).

Future work will focus on finding improved capacity upper bounds, for example using the duality framework in~\cite{lapidoth2003duality,lapidoth2009optical}, and on extending the analysis to multiple-antenna systems.




\section*{Acknowledgment}

This work was supported by the ERC under Grant Agreement 101125691.

\appendices

\section{Derivation of the Effective Noise Distribution}
\label{app:noise_dist}
Let \(\bar{W} = W + W_q\), where \(W \sim \mathcal{N}(0,\sigma^2)\) and 
\(W_q \sim \mathcal{U}(-\Delta/2,\Delta/2)\) are mutually independent.
\subsection{Probability Density Function}
The PDF of \(\bar{W}\) is given by the convolution
\begin{equation}
    f_{\bar{W}}(w)
    = \frac{1}{\Delta} \int_{-\Delta/2}^{\Delta/2}
    \frac{1}{\sigma}\,\phi\!\left(\frac{w - z}{\sigma}\right) dz,
\end{equation}
where \(\phi(\cdot)\) and \(Q(\cdot)\) denote the standard Gaussian PDF and the Gaussian \(Q\)-function, respectively. This evaluates to
\begin{equation}
    f_{\bar{W}}(w)
    = \frac{1}{\Delta}
    \left[
        Q\!\left(\frac{w - \frac{\Delta}{2}}{\sigma}\right)
        -
        Q\!\left(\frac{w + \frac{\Delta}{2}}{\sigma}\right)
    \right].
    \label{eq:fwbar_app}
\end{equation}
\subsection{Cumulative Distribution Function}
The CDF of \(\bar{W}\) is given by
\begin{align}
    F_{\bar W}(w)
    &=
    1+\frac{1}{\Delta}
    \Bigg[
    \left(w-\frac{\Delta}{2}\right)
    Q\!\left(\frac{w-\frac{\Delta}{2}}{\sigma}\right)
    \nonumber\\
    &\quad
    -
    \left(w+\frac{\Delta}{2}\right)
    Q\!\left(\frac{w+\frac{\Delta}{2}}{\sigma}\right)
    \nonumber\\
    &\quad
    +
    \sigma
    \left(
    \phi\!\left(\frac{w+\frac{\Delta}{2}}{\sigma}\right)
    -
    \phi\!\left(\frac{w-\frac{\Delta}{2}}{\sigma}\right)
    \right)
    \Bigg].
    \label{eq:Fwbar_app}
\end{align}

\section{Proof of Theorem~\ref{thm:epi_lb}}
\label{app:epi_lb}

By the entropy power inequality,
\begin{align}
I(X;Y) &= h(X+\bar W) - h(\bar W)\\
&\ge
\frac{1}{2}\log\!\left(e^{2h(X)} + e^{2h(\bar W)}\right) - h(\bar W).
\end{align}
Taking the supremum over \(P_X \in \mathcal{F}\) yields~\eqref{eq:capacity_LB_EPI_main} with \(h_{\max}(A,P_0)\) as defined in~\eqref{eq:hmax}.

It remains to evaluate \(h_{\max}\). By~\cite[Ch.~11]{coverthomas}, the maximum-entropy density under~\eqref{eq:peak_constraint} and~\eqref{eq:avg_constraint} has the form
\begin{equation}
f_X^\star(x) \propto e^{-\mu x^2}, \qquad |x|\le A,
\end{equation}
for some parameter \(\mu \ge 0\). If the power constraint is inactive \((\mu = 0)\), the density is uniform on \([-A,A]\) and \(h_{\max} = \log(2A)\), which is feasible when \(P_0 \ge A^2/3\). If it is active \((\mu > 0)\), the density is a truncated Gaussian of the form
\begin{equation}
    f_X^\star(x)
    =
    \frac{\sqrt{\mu^{\star}}}
    {\sqrt{\pi}\Bigl(1 - 2Q\!\bigl(\sqrt{2\mu^{\star}}\,A\bigr)\Bigr)}
    e^{-\mu^{\star} x^2},
    \qquad |x|\le A,
\end{equation}
and the parameter \(\mu^\star > 0\) is determined by \(\mathbb{E}[(X^\star)^2] = P_0\), which gives~\eqref{eq:mu_star_theorem} and determines \(\mu^\star\) uniquely. Together the two cases give~\eqref{eq:hmax2}.






\bibliographystyle{IEEEtran}
\bibliography{refs}
\vspace{12pt}

\end{document}